\documentclass[preprint,12pt]{elsarticle}

\usepackage{amsmath,amsthm,amsfonts,amssymb,amscd}
\usepackage{mathrsfs}
\usepackage{tikz}
\usepackage[utf8]{inputenc}
\usepackage{natbib}

\newtheorem{theorem}{Theorem}
\newtheorem{lemma}[theorem]{Lemma}
\newtheorem{corollary}[theorem]{Corollary}

\newtheorem{example}[theorem]{Example}

\begin{document}


\begin{frontmatter}

\title{Characterizations of graph classes via convex geometries: a survey\footnote{Mitre C. Dourado is partially supported by CNPq, Brazil (305404/2020-2 and 403601/2023-1) and FAPERJ, Brazil (SEI-260003/015314/2021); Fábio Protti is partially supported by CNPq (304117/2019-6) and FAPERJ (E-26/201.083/2021) ; Rudini Sampaio is partially supported by CNPq (404479/2023-5 and 311070/2022-1) and CAPES (88881.712024/2022-01).}}


\author{Mitre C. Dourado}
\address{Instituto de Computa\c c\~ao, Universidade Federal do Rio de Janeiro, Brazil}
\ead{mitre@ic.ufrj.br}

\author{Marisa Gutierrez}
\address{CMaLP, Facultad de Ciencias Exactas, Universidad Nacional de La Plata, Argentina}
\ead{marisa@mate.unlp.edu.ar}

\author{F\'abio Protti}
\address{Instituto de Computa\c c\~ao, Universidade Federal Fluminense, Brazil}
\ead{fabio@ic.uff.br }

\author{Rudini Sampaio}
\address{Departamento de Computa\c c\~ao, Universidade Federal do Cear\'a, Brazil}
\ead{rudini@dc.ufc.br}

\author{Silvia Tondato}
\address{CMaLP, Facultad de Ciencias Exactas, Universidad Nacional de La Plata, Argentina}
\ead{tondato@mate.unlp.edu.ar}

\begin{abstract}
Graph convexity has been used as an important tool to better understand the structure of classes of graphs. Many studies are devoted to determine if a graph equipped with a convexity is a {\em convex geometry}. In this work, we survey results on characterizations of well-known classes of graphs via convex geometries. We also give some contributions to this subject.
\end{abstract}


\end{frontmatter}

\section{Introduction}

A \textit{convexity} on a nonempty set $V$ is a family $\mathscr{C}$ of subsets of $V$ (called {\em convex sets}) such that: (a) $\emptyset$ and $V$ are convex sets; (b) the intersection of convex sets is a convex set. If $V=V(G)$ for some graph $G$, then $\mathscr{C}$ is a \textit{graph convexity} of $G$. The structure of a convexity is based on the way convex sets are defined. In the context of graph convexities, there are many examples in the literature where convex sets are defined over a {\em path system}, typically according to the following rule: fix a family ${\mathscr P}_G$ of paths (or, more generally, walks) of a graph $G$, and say that a set $S\subseteq V(G)$ is convex if no vertex $x\notin S$ lies in a member of ${\mathscr P}_G$ that starts and ends at two vertices of $S$. The most natural example is the so-called {\em geodesic convexity}~\cite{pelayo}, for which ${\mathscr P}_G$ is the family of all the {\em shortest paths} of $G$. Other important examples are: the {\em monophonic convexity}~\cite{dourado-et-al,duchet}, the {\em $m^3$-convexity}~\cite{dragan-et-al}, the {\em $l^k$-convexity}~\cite{gpt}, the {\em toll convexity}~\cite{alcon-et-al}, and the {\em weakly toll convexity}~\cite{gutierrez-tondato}, defined over induced paths, induced paths of length at least three, induced paths of length at most $k$, tolled walks, and weakly toll walks, respectively. Other rules to define convex sets, not based on path systems, have also been investigated; an important example comes from the so-called {\em Steiner convexity}, introduced in~\cite{caceres}. In this case, a set $S\subseteq V(G)$ is {\em Steiner convex} if, for every $S'\subseteq S$, the vertices of any Steiner tree with terminal set $S'$ belong to $S$.

Graph convexities have been studied in many contexts. A strong direction of research has focused on determining convexity invariants, such as the {\em hull number}, the {\em interval number}, and the {\em convexity number}, among others. A major reference work by Pelayo~\cite{pelayo} gives an extensive overview on convexity invariants, applied to the case of the geodesic convexity.

Other studies are devoted to determine if a graph convexity is a {\em convex geometry}. This line of research is closely linked to the original idea of defining a combinatorial abstraction of convex sets in Geometry. The convex sets in the plane are those subsets $S$ such that no point outside $S$ lies on a line segment with endpoints in $S$. The {\em convex hull} of a set of points $A$ in the plane is the smallest convex set that contains $A$. A point $x$ of a convex set $S$ is an {\em extreme point} of $S$ if $S\setminus\{x\}$ is also a convex set, that is, if there is no line segment with endpoints in $S\setminus\{x\}$ containing $x$. For instance, an $n$-polygon (considered as a closed region) is clearly a convex set in the plane, and its extreme points are precisely its $n$ vertices\footnote{Here, the term {\em vertices} is used in the geometric sense.}. In addition, the polygon is the convex hull of its $n$ extreme points.  Another example is a circle (closed circular region); it is clearly a convex set in the plane, and all the points lying on the circle border are its extreme points. Again, the circle is the convex hull of such points.

The above concepts can be transferred to the combinatorial field in a natural way. We refer the reader to~\cite{farber-jamison}. Let $G$ be a graph and let $\mathscr{C}$ be a convexity of $G$. Given a set $S\subseteq V(G)$, the smallest set $H\in\mathscr{C}$ containing $S$ is called the \textit{convex hull} of $S$. A vertex $x$ of a convex set $S$ is an \textit{extreme vertex} of $S$ if $S\backslash\{x\}$ is also convex. The convexity $\mathscr{C}$ is a \textit{convex geometry} if it satisfies the \textit{Krein-Milman} (or \textit{Minkowski-Krein-Milman}) property~\cite{krein-milman,pelayo-survey}: {\em Every convex set is the convex hull of its extreme vertices.}  The main question dealt with in this survey is: by fixing a rule $r$ to define the convex sets (e.g., a rule based on some path system), determine the class of graphs whose $r$-convexities are convex geometries. For instance, by fixing induced paths, we can obtain the following characterization: a graph $G$ is chordal if and only if the monophonic convexity of $G$ is a convex geometry~\cite{farber-jamison}. Ptolemaic graphs, interval graphs, proper interval graphs, weak bipolarizable graphs, and 3-fan-free chordal graphs can also be characterized in this way by considering, respectively, the geodesic convexity~\cite{farber-jamison}, the toll covexity~\cite{alcon-et-al}, the weakly toll convexity~\cite{gutierrez-tondato}, the $m^3$-convexity~\cite{dragan-et-al}, and the Steiner convexity~\cite{caceres}. In~\cite{gpt}, a characterization of graphs with $l^k$-convexities that are convex geometries are studied. Section 3 discusses in detail most of these characterizations.

Convex geometries are equivalent, by complementation, to {\em an\-ti\-ma\-troids} (see \cite{korte-lovasz-schrader}). An antimatroid consists of a finite family ${\mathcal F}$ of finite sets (called {\em feasible sets}) such that: (a) ${\mathcal F}$ is closed under unions; (b) if $S$ is a nonempty feasible set, then there is $x\in S$ for which $S\setminus \{x\}$ is also a feasible set (that is, ${\mathcal F}$ is an {\em accessible set system}). If $U$ is the union of the sets in ${\mathcal F}$, then the family of complementary sets ${\mathscr C}=\{U\setminus S\mid S\in {\mathcal F}\}$ is a convex geometry. It is not difficult to see that ${\mathscr C}$ is a convexity; in addition, it satisfies the {\em anti-exchange property}~\cite{edelman,korte-lovasz-schrader}, which is complementary to the accessibility property of antimatroids (property (b) above). In Section 2, we present in more detail the anti-exchange property, which will be useful to derive some results in Section 4.

This survey is organized as follows. Section 2 provides the necessary background. In Section 3, we review the main results in the literature on characterizations of graph classes via convex geometries; chordal, Ptolemaic, strongly chordal, interval, proper interval graphs, among other classes, are dealt with. In Section 4 we present some contributions; namely, we show that forests, forests of stars, cographs, bipartite graphs, planar graphs, and triangle-free graphs can be well characterized via convex geometries. Our concluding remarks are presented in Section 5.

\section{Preliminaries}

All the graphs in this work are connected, unless otherwise stated. A path (resp., cycle) in a graph $G$ is {\em induced} if there is no edge of $G$ linking two nonconsecutive vertices of it. We denote by $P_n$ (resp., $C_n$) the induced path (resp., cycle) with $n$ vertices. For a vertex $v\in V(G)$, let $N(v)$ denote the {\em open neighborhood} (or simply {\em neighborhood}) of $v$ in $G$. In addition, let $N[v]$ denote the {\em closed neighborhood} of $v$ in $G$, where $N[v]=N(v)\cup\{v\}$. If $N[v]$ is a clique, then $v$ is said to be a {\em simplicial vertex}. For a set $S\subseteq V(G)$, we denote by $G[S]$ the subgraph of $G$ {\em induced} by $S$, defined as follows: $V(G[S])=S$ and $E(G[S])=\{xy\in E(G)\mid x,y\in S\}$. For a graph $F$, we say that $G$ is {\em $F$-free} if $G$ does not contain $F$ as an induced subgraph. For a family ${\cal F}=\{G_1,G_2,\ldots,G_k\}$ of prefixed graphs, we say that $G$ is {\em $(G_1,G_2,\ldots,G_k)$-free} (or simply {\em ${\cal F}$-free}) if $G$ is $G_i$-free for $1\leq i\leq k$. For ${\cal C}=\{C_k \mid k\geq 4\}$, $G$ is said to be a {\em chordal graph} if $G$ is ${\cal C}$-free.    

Let ${\mathscr P}$ be a function (called {\em path system}) that maps each connected graph $G$ to a collection ${\mathscr P}_G$ of walks of $G$. For instance, ${\mathscr P}$ can map $G$ to its shortest paths, induced paths, tolled walks etc. The members of ${\mathscr P}_G$ are generically called {\em p-walks} of $G$. All the path systems dealt with in this paper satisfy the following property: For every $G$, all the walks of $G$ formed by a single vertex or two adjacent vertices are p-walks.

The {\em interval} $I(u,v)$ of $u,v\in V(G)$ is the set of all vertices that lie in a p-walk from $u$ to $v$. For $S\subseteq V(G)$, we define $I(S)=\cup_{u,v\in S} I(u,v)$. A set $S\subseteq V(G)$ is {\em convex} if $I(S)=S$. Let ${\mathscr C}=\{S\subseteq V(G): I(S)=S\}$. Clearly, ${\mathscr C}$ is a convexity of $G$, defined over the collection ${\mathscr P}_G$ of p-walks.

A vertex $x$ of a convex set $S\subseteq V(G)$ is an {\em extreme vertex} of $S$ if $S\setminus\{x\}$ is also a convex set. This definition implies that: (a) no p-walk in $G[S]$ contains $x$ as an internal vertex; (b) $x\notin I(u,v)$ for any pair of distinct vertices $u,v\in S$. The set of extreme vertices of $S$ is denoted by ${\mathit ext}(S)$. The {\em convex hull} of a set $S \subseteq V(G)$, denoted by $H(S)$, is the smallest set convex set of $G$ containing $S$. It can be shown that $H(S)$ can be iteratively constructed by successive applications of the interval operator, as follows: starting with a set $S\subseteq V(G)$, define $S_0=S$ and $S_i=I(S_{i-1}), i\geq 1$. Then, there exists $k\geq 1$ such that $S_k=S_{k-1}=H(S)$.

The concept of interval is the combinatorial analog of a line segment between two points in the plane. Just as the geometric convex hull of a set of points in the plane can be built using line segments, the convex hull of a set $S$ of vertices in a graph uses the notion of interval for the iterative construction described above.




We say that a convexity ${\mathscr C}$ of a graph $G$ is a \textit{convex geometry} (or that it is {\em geometric}) if it satisfies the following property: for every $S\in{\mathscr C}$, it holds that $S=H({\mathit ext}(S))$. This is called the \textit{Krein-Milman} property~\cite{krein-milman}. Alternatively, the convexity ${\mathscr C}$ is a convex geometry if the following axiom is satisfied:

\medskip

\noindent {\bf Anti-exchange property}~\cite{edelman,korte-lovasz-schrader}\\
If $y,z \notin H(S)$ and $z\in H(S\cup\{y\})$, then $y\notin H(S\cup\{z\})$.

\medskip

The anti-exchange property is a combinatorial abstraction of the usual convex hull operator in the plane: for two points $y$ and $z$ not in the convex hull of $S$, if $z$ is in the convex hull of $S\cup\{y\}$, then $y$ is outside the convex hull of $S\cup\{z\}$. See Figure 1.

\begin{figure}[h]
\begin{center}
\begin{tikzpicture}
\filldraw[color=black, fill=gray!10] (0, 0) rectangle (2, 2);
\filldraw[black] (0,0) circle (2pt) node[anchor=south west]{$a$};
\filldraw[black] (0,2) circle (2pt) node[anchor=north west]{$b$};
\filldraw[black] (2,2) circle (2pt) node[anchor=north east]{$c$};
\filldraw[black] (2,0) circle (2pt) node[anchor=south east]{$d$};
\filldraw[black] (3,1) circle (2pt) node[anchor=west]{$z$};
\filldraw[black] (4,1) circle (2pt) node[anchor=west]{$y$};
\node at (1,1) {$H(S)$};
\draw[dashed] (2,2) -- (4,1);
\draw[dashed] (2,0) -- (4,1);
\end{tikzpicture}
\caption{Anti-exchange property applied to $S=\{a,b,c,d\}$.}
\end{center}
\end{figure}

\section{An overview of graph classes and convex geometries}

In this section, we briefly review the main existing results on convex geometries in the literature on graph convexity. Chordal graphs play in important role in this section. For basic facts on chordal graphs, see~\cite{golumbic80}.

\subsection{Monophonic convexity}

Let ${\mathscr P}^m$ be the path system that maps each graph $G$ to its induced paths. The convexity defined over the induced paths of $G$ is called the {\em monophonic convexity} of $G$. It is easy to see that $v$ is an extreme vertex of a (monophonically) convex set $S\subseteq V(G)$ if and only if $v$ is a simplicial vertex of $G[S]$.

Suppose that $G$ contains an induced cycle $C$ with at least four vertices. Clearly, $S=H(V(C))$ is convex and contains no simplicial vertices, i.e., ${\mathit ext}(S)=\emptyset$. Thus, the Krein-Milman property fails in this case. In other words, if the monophonic convexity of $G$ is a convex geometry, then $G$ is chordal. On the other hand, if $G$ is chordal, then every nonsimplicial vertex of $G$ lies on an induced path between two simplicial vertices~\cite{farber-jamison}. As chordality is an hereditary property, this implies that every convex subset of vertices of $G$ is indeed the convex hull of its extreme vertices. Hence:

\begin{theorem} \label{the:monophonic} {\em \cite{farber-jamison}}
A graph $G$ is chordal if and only if the monophonic convexity of $G$ is a convex geometry.
\end{theorem}

\subsection{Geodesic convexity}

A path $P$ between two vertices $u$ and $v$ in a graph $G$ is a {\em shortest path} if the length of $P$ is equal to the distance between $u$ and $v$ in $G$. Let ${\mathscr P}^g$ be the path system that maps each graph $G$ to its shortest paths. The convexity defined over the shortest paths of $G$ is the {\em geodesic convexity} of $G$.

As in the previous subsection, it is easy to see that $v$ is an extreme vertex of a (geodesically) convex set $S\subseteq V(G)$ if and only if $v$ is a simplicial vertex of $G[S]$. In addition, if the geodesic convexity of $G$ is a convex geometry, then $G$ must be clearly chordal. However, it is not true that, for every chordal graph $G$, the geodesic convexity of $G$ is a convex geometry. For instance, consider the graph $G'$ (called {\em gem}) with $V(G')=\{a,b,c,d,e\}$ and $E(G')=\{ab,bc,cd,de,ae,be,ce,de\}$. See Figure~\ref{fig:gem}. Assume that $G'$ is an induced subgraph of a chordal graph $G$. Clearly, the extreme vertices of $Y=H(V(G'))$ are in $\{a,d\}$. However, $H(\{a,d\})\neq Y$.

\begin{figure}[h]
\begin{center}
\begin{tikzpicture}
\filldraw[black] (1,2) circle (2pt) node[anchor=south]{$e$};
\filldraw[black] (0,1.1) circle (2pt)node[anchor=east]{$d$};
\filldraw[black] (2,1.1) circle (2pt)node[anchor=west]{$a$};
\filldraw[black] (0.4,0) circle (2pt)node[anchor=east]{$c$};
\filldraw[black] (1.6,0) circle (2pt)node[anchor=west]{$b$};
\draw (1,2) -- (0,1.1);
\draw (1,2) -- (2,1.1);
\draw (0,1.1) -- (0.4,0);
\draw (2,1.1) -- (1.6,0);
\draw (0.4,0) -- (1.6,0);
\draw (1,2) -- (0.4,0);
\draw (1,2) -- (1.6,0);
\end{tikzpicture}
\caption{A gem $G'$ with $V(G')=\{a,b,c,d,e\}$.}\label{fig:gem}
\end{center}
\end{figure}

A graph $G$ is {\em Ptolemaic} if it is chordal and gem-free. The arguments in the previous paragraph tell us that if the geodesic convexity of $G$ is a convex geometry, then $G$ is a Ptolemaic graph. Conversely, if $G$ is a Ptolemaic graph, then it can be proved that every induced path of $G$ is a shortest path~\cite{howorka}, i.e., the monophonic and the geodesic convexities of $G$ coincide in this case. Therefore, in a Ptolemaic graph, every nonsimplicial vertex of $G$ lies on a {\em shortest} path between two simplicial vertices. Consequently, every convex subset of vertices of $G$ is indeed the convex hull of its extreme vertices. Hence:

\begin{theorem} {\em \cite{farber-jamison}}\label{thm:ptolemaic}
A graph $G$ is Ptolemaic if and only the geodesic convexity of $G$ is a convex geometry.
\end{theorem}

\subsection{Strong convexity}

We say that a path $P=u_0u_1\ldots u_n$ is {\em even-chorded} if it has no odd chord  (an edge that connects two vertices that are an odd distance $d>1$ apart from each other in a path or cycle) and, in addition, neither $u_0$ nor $u_n$ lies in a chord of $P$. Let ${\mathscr P}^s$ be the path system that maps a graph $G$ to its even-chorded paths. The convexity defined over such paths is called the {\em strong convexity} of $G$, and the associated convex sets are called {\em strong convex sets}.

An {\em even cycle} is a cycle with an even number of vertices. A graph is {\em strongly chordal}~\cite{farber} if it is chordal and every of its even cycles with at least six vertices has an odd chord. A vertex of a graph is {\em simple} if the neighborhoods of its neighbors form a nested family of sets. Note that a simple vertex is simplicial, but not conversely. In~\cite{farber}, the following characterization of strongly chordal graphs is given:

\begin{theorem}\label{thm:sc} {\em \cite{farber}}
A graph $G$ is a strongly chordal graph if and only if every induced subgraph of $G$ contains a simple vertex.
\end{theorem}

Suppose that the strong convexity of $G$ is a convex geometry. Since induced paths are even-chorded, an extreme vertex $v$ of an strong convex set $S$ must be simplicial in $G[S]$. This implies that $G$ is chordal, and, in this case, we can prove the following lemma:

\begin{lemma}\label{lem:sc} {\em \cite{farber-jamison}}
Let $G$ be a chordal graph. Then, a vertex $v\in V(G)$ is an extreme vertex of an strong convex set $S$ if and only if $v$ is a simple vertex in $G[S]$.
\end{lemma}

By Theorem~\ref{thm:sc} and Lemma~\ref{lem:sc}, $G$ is a strongly chordal graph. Conversely, suppose that $G$ is a strongly chordal graph. In~\cite{farber-jamison}, the following result is proved:

\begin{lemma}\label{lem:sc2} {\em \cite{farber-jamison}}
In a strongly chordal graph, every nonsimple vertex lies in an even-chorded path between simple vertices.
\end{lemma}

By Lemma~\ref{lem:sc2}, $V(G)$ is indeed the convex hull of its extreme vertices. But since strongly chordal graphs are hereditary, we have:

\begin{theorem} {\em \cite{farber-jamison}}
A graph $G$ is strongly chordal if and only if the strong convexity of $G$ is a convex geometry.
\end{theorem}

\subsection{$m^3$-convexity}

Let ${\mathscr P}^{\geq 3}$ be the path system that maps a graph $G$ to its induced paths of length at least three. The convexity defined over such paths is called the {\em $m^3$-convexity} of $G$, and the associated convex sets are called {\em $m^3$-convex sets}. Note that an $m^3$-convex set does not necessarily induce a connected subgraph.

An alternative definition of ``simplicial'' is: $v$ is a simplicial vertex iff it is not a midpoint of an induced $P_3$. In~\cite{jamison-olariu}, this concept is relaxed as follows: a vertex is {\em semisimplicial} if it is not an internal vertex of an induced $P_4$, and {\em nonsemisimplicial} otherwise. Clearly, a vertex $v$ is an extreme vertex of an $m^3$-convex set $S\subseteq V(G)$ if and only if $v$ is semisimplicial in $G[S]$.

A graph $G$ is {\em weak bipolarizable}~\cite{olariu} if $G$ contains no hole, house, domino, or $A$ as an induced subgraph. See Figure~\ref{fig:hhda}. Weak bipolarizable graphs are also called HHDA-free graphs. 

\begin{figure}[h]
\begin{center}
\begin{tikzpicture}[scale=0.9]
\filldraw[black] (1,2) circle (2pt);
\filldraw[black] (0,1.1) circle (2pt);
\filldraw[black] (2,1.1) circle (2pt);
\filldraw[black] (0.4,0) circle (2pt);
\filldraw[black] (1.6,0) circle (2pt);
\draw (1,2) -- (0,1.1);
\draw (1,2) -- (2,1.1);
\draw (0,1.1) -- (0.4,0);
\draw (2,1.1) -- (1.6,0);
\draw[dashed] (0.4,0) -- (1.6,0);
\filldraw[black] (5,2) circle (2pt);
\filldraw[black] (4,1.1) circle (2pt);
\filldraw[black] (6,1.1) circle (2pt);
\filldraw[black] (4.4,0) circle (2pt);
\filldraw[black] (5.6,0) circle (2pt);
\draw (5,2) -- (4,1.1);
\draw (5,2) -- (6,1.1);
\draw (4,1.1) -- (4.4,0);
\draw (6,1.1) -- (5.6,0);
\draw (4.4,0) -- (5.6,0);
\draw (4,1.1) -- (6, 1.1);
\filldraw[black] (8,2) circle (2pt);
\filldraw[black] (10,2) circle (2pt);
\filldraw[black] (8,1) circle (2pt);
\filldraw[black] (10,1) circle (2pt);
\filldraw[black] (8,0) circle (2pt);
\filldraw[black] (10,0) circle (2pt);
\draw (8,2) -- (10,2) -- (10,1) -- (10,0) -- (8,0) -- (8,1) -- (8,2);
\draw (8,1) -- (10,1);
\filldraw[black] (12,2) circle (2pt);
\filldraw[black] (14,2) circle (2pt);
\filldraw[black] (12,1) circle (2pt);
\filldraw[black] (14,1) circle (2pt);
\filldraw[black] (12,0) circle (2pt);
\filldraw[black] (14,0) circle (2pt);
\draw (12,0) -- (12,1) -- (12,2) -- (14,2) -- (14,1) -- (14,0);
\draw (12,1) -- (14,1);
\end{tikzpicture}
\caption{From left to right: hole, house, domino, $A$.}\label{fig:hhda}
\end{center}
\end{figure}

Using arguments similar to those presented in the preceding subsections, it can be shown if the $m^3$-convexity of $G$ is a convex geometry, then $G$ is a weak bipolarizable graph. Suppose the $m^3$-convexity of $G$ is a convex geometry and $G$ contains, for instance, a house $G'$ as an induced subgraph. Then, the only possible semisimplicial vertex in the subgraph of $G$ induced by $Y=H(V(G'))$ is vertex $e$, i.e., ${\mathit ext}(Y)\subseteq\{e\}$. This implies $H({\mathit ext}(Y))\neq Y$, a contradiction. Analogously, $G$ cannot contain a hole, a domino, or an $A$ as an induced subgraph. Thus, $G$ is weak bipolarizable. Conversely, if $G$ is weak bipolarizable, it holds that every nonsemisimplicial vertex lies on an induced path of length at least three between semisimplicial vertices~\cite{dragan-et-al}. We then have the following theorem:

\begin{theorem} {\em \cite{dragan-et-al}}
A graph $G$ is weak bipolarizable if and only if the $m^3$-convexity of $G$ is a convex geometry.
\end{theorem}

\subsection{Steiner convexity}

Let $S\subseteq V(G)$. A {\em Steiner tree for $S$} is a connected subgraph $T$ of $G$ such that $S\subseteq V(T)$ and $|E(T)|$ is minimum. The {\em Steiner interval} $I(S)$ is the set formed by all vertices of $G$ that lie in some Steiner tree for $S$. We say that $S$ is {\em Steiner convex} if, for every $S'\subseteq S$, it holds that $I(S')\subseteq S$. The {\em Steiner convexity} of $G$ is the family formed by the Steiner convex sets of $G$. Important references on the Steiner convexity are~\cite{caceres,dourado-oliveira-protti}.  

In~\cite{caceres}, the authors prove that: (a) the extreme vertices of a Steiner convex set $S$ of $G$ are the simplicial vertices of $G[S]$; (b) in a (house,hole,domino)-free graph $G$, the Steiner convexity and the geodesic convexity coincide. Such facts are the main ingredients to prove the following characterization:

\begin{theorem} {\em \cite{caceres}}
A graph $G$ is Ptolemaic if and only if the Steiner convexity of $G$ is a convex geometry.
\end{theorem}

\begin{proof}
Suppose that $G$ is Ptolemaic, i.e., $G$ is chordal and gem-free. Then $G$ is clearly (house,hole,domino)-free, so the Steiner convexity of $G$ coincides with the geodesic convexity of $G$, which is a convex geometry by Theorem~\ref{thm:ptolemaic}.

Conversely, suppose that the Steiner convexity of $G$ is a convex geometry. We proceed with arguments that we have used before. If $G$ contains an induced gem $G'$, then, clearly, the extreme vertices of $Y=H(V(G'))$ are in $\{a,d\}$ (see Figure~\ref{fig:gem}). However, $H(\{a,d\})\neq Y$. Thus, $G$ is gem-free. Likewise, if $G$ contains an induced cycle $C$ with at least four vertices then $W=H(V(C))$ is Steiner convex and contains no simplicial vertices, i.e., ${\mathit ext}(W)=\emptyset$. Thus, $G$ is chordal. 
\end{proof}

In~\cite{caceres-oellermann}, an alternative way of defining Steiner convexity is proposed, by looking at subsets with a prescribed size. We say that $S$ is {\em $k$-Steiner convex} or $g_k$-convex if, for every $S'\subseteq S$ with $|S'|=k$, it holds that $I(S')\subseteq S$. The convexity associated with the $g_k$-convex sets of $G$ is called the {\em $k$-Steiner convexity} of $G$. Note that the $2$-Steiner convexity of $G$ is precisely the geodesic convexity of $G$. 

Let us study the case $k=3$. In~\cite{caceres-oellermann}, the extreme vertices of $g_3$-convex sets are characterized: a vertex $x$ is an extreme vertex of a $g_3$-convex set $S$ of $G$ if and only if $x$ is {\em not} the central vertex of an induced claw, paw, or $P_4$ contained in $G[S]$. See Figure~\ref{fig:central}.

\begin{figure}[h]
\begin{center}
\begin{tikzpicture}
\filldraw[black] (1,1) circle (2pt) node[anchor=south]{$\downarrow$};
\filldraw[black] (0,0) circle (2pt);
\filldraw[black] (1,0) circle (2pt);
\filldraw[black] (2,0) circle (2pt);
\filldraw[black] (5,1) circle (2pt) node[anchor=south]{$\downarrow$};
\filldraw[black] (4,0) circle (2pt);
\filldraw[black] (5,0) circle (2pt);
\filldraw[black] (6,0) circle (2pt);
\filldraw[black] (8,0.5) circle (2pt);
\filldraw[black] (9,0.5) circle (2pt) node[anchor=south]{$\downarrow$};
\filldraw[black] (10,0.5) circle (2pt)node[anchor=south]{$\downarrow$};
\filldraw[black] (11,0.5) circle (2pt);
\draw (1,1) -- (0,0);
\draw (1,1) -- (1,0);
\draw (1,1) -- (2,0);
\draw (5,1) -- (4,0);
\draw (5,1) -- (5,0);
\draw (5,1) -- (6,0);
\draw (4,0) -- (5,0);
\draw (8,0.5) -- (9,0.5) -- (10,0.5) -- (11,0.5);
\end{tikzpicture}
\caption{From left to right: claw, paw, $P_4$. Central vertices are indicated by arrows.}\label{fig:central}
\end{center}
\end{figure}

Suppose that the $3$-Steiner convexity of $G$ is a convex geometry. Suppose also that $G$ contains an induced graph $P$ isomorphic to a $P_4=abcd$, where $a$ and $d$ have degree one in $P$. Observe that $a$ and $d$ are the only extreme vertices of $P$. Thus, ${\mathit ext}(V(P))$ is a set $C\subseteq\{a,d\}$. However, $H(C)=C\neq H(\{a,b,c,d\})$, that is, $H(\{a,b,c,d\})$ is a 3-Steiner convex set that is not the convex hull of its extreme points. This is a contradiction. Hence, $G$ is $P_4$-free.  

Now consider the four graphs $G_1,G_2,G_3,G_4$ depicted in Figure~\ref{fig:replicated}. Note that ${\mathit ext}(G_i)=\emptyset, 1\leq i\leq 4$. For instance, it can be checked by inspection that every vertex in $G_1$ is a central vertex of a claw. A similar behavior occurs for $G_2$, $G_3$, and $G_4$. This implies that $Y=H(G_i)$ is a 3-Steiner convex set and ${\mathit ext}(Y)=\emptyset$, a contradiction. Hence, $G$ is $(G_1,G_2,G_3,G_4)$-free. 

On the other hand, in~\cite{caceres-oellermann} the authors show that if $G$ is $(P_4,G_1,G_2,G_3,G_4)$-free then the $3$-Steiner convexity of $G$ is a convex geometry; this proof is a bit laborious, and so it is omitted here.

\begin{theorem} {\em \cite{caceres-oellermann}}
The $3$-Steiner convexity of $G$ is a convex geometry if and only if $G$ is $(P_4,G_1,G_2,G_3,G_4)$-free.
\end{theorem}

\begin{figure}[h]
\begin{center}
\begin{tikzpicture}[scale=0.9]
\filldraw[black] (1,2) circle (2pt) node[anchor=south]{$G_1$};
\filldraw[black] (0,1.1) circle (2pt);
\filldraw[black] (1,1.1) circle (2pt);
\filldraw[black] (2,1.1) circle (2pt);
\filldraw[black] (0.4,0) circle (2pt);
\filldraw[black] (1.6,0) circle (2pt);
\draw (1,2) -- (0,1.1);
\draw (1,2) -- (1,1.1);
\draw (1,2) -- (2,1.1);
\draw (0.4,0) -- (0,1.1);
\draw (0.4,0) -- (1,1.1);
\draw (0.4,0) -- (2,1.1);
\draw (1.6,0) -- (0,1.1);
\draw (1.6,0) -- (1,1.1);
\draw (1.6,0) -- (2,1.1);
\filldraw[black] (5,2) circle (2pt) node[anchor=south]{$G_2$};
\filldraw[black] (4,1.1) circle (2pt);
\filldraw[black] (5,1.1) circle (2pt);
\filldraw[black] (6,1.1) circle (2pt);
\filldraw[black] (4.4,0) circle (2pt);
\filldraw[black] (5.6,0) circle (2pt);
\draw (5,2) -- (4,1.1);
\draw (5,2) -- (5,1.1);
\draw (5,2) -- (6,1.1);
\draw (4.4,0) -- (4,1.1);
\draw (4.4,0) -- (5,1.1);
\draw (4.4,0) -- (6,1.1);
\draw (5.6,0) -- (4,1.1);
\draw (5.6,0) -- (5,1.1);
\draw (5.6,0) -- (6,1.1);
\draw (5,1.1) -- (4,1.1);
\filldraw[black] (9,2) circle (2pt) node[anchor=south]{$G_3$};
\filldraw[black] (8,1.1) circle (2pt);
\filldraw[black] (9,1.1) circle (2pt);
\filldraw[black] (10,1.1) circle (2pt);
\filldraw[black] (8.4,0) circle (2pt);
\filldraw[black] (9.6,0) circle (2pt);
\draw (9,2) -- (8,1.1);
\draw (9,2) -- (9,1.1);
\draw (9,2) -- (10,1.1);
\draw (8.4,0) -- (8,1.1);
\draw (8.4,0) -- (9,1.1);
\draw (8.4,0) -- (10,1.1);
\draw (9.6,0) -- (8,1.1);
\draw (9.6,0) -- (9,1.1);
\draw (9.6,0) -- (10,1.1);
\draw (9.6,0) -- (8.4,0);
\filldraw[black] (13,2) circle (2pt) node[anchor=south]{$G_4$};
\filldraw[black] (12,1.1) circle (2pt);
\filldraw[black] (13,1.1) circle (2pt);
\filldraw[black] (14,1.1) circle (2pt);
\filldraw[black] (12.4,0) circle (2pt);
\filldraw[black] (13.6,0) circle (2pt);
\draw (13,2) -- (12,1.1);
\draw (13,2) -- (13,1.1);
\draw (13,2) -- (14,1.1);
\draw (12.4,0) -- (12,1.1);
\draw (12.4,0) -- (13,1.1);
\draw (12.4,0) -- (14,1.1);
\draw (13.6,0) -- (12,1.1);
\draw (13.6,0) -- (13,1.1);
\draw (13.6,0) -- (14,1.1);
\draw (13,1.1) -- (12,1.1);
\draw (13.6,0) -- (12.4,0);
\end{tikzpicture}
\caption{Graphs $G_1, G_2, G_3, G_4$.}\label{fig:replicated}
\end{center}
\end{figure}

Let ${\cal R}=\{k_1, k_2, \ldots, k_t\}$ be a set of positive integers with $k_1\leq\cdots\leq k_t$ and $k_1\in\{2,3\}$. Say that a set $S\subseteq V(G)$ is {\em ${\cal R}$-Steiner convex} if $S$ is $k_i$-Steiner convex for $1\leq i\leq t$. The convexity associated with ${\cal R}$-Steiner convex sets is the {\em ${\cal R}$-Steiner convexity}. Interestingly, the class of graphs for which the ${\cal R}$-Steiner convexity is a convex geometry is exactly the same as the class for which the $k_1$-Steiner convexity is a convex geometry:

\begin{theorem} {\em \cite{caceres-oellermann}}
The following statements are true:\\
\begin{enumerate}
\item If $k=2$, then the ${\cal R}$-Steiner convexity is a convex geometry if and only if $G$ is Ptolemaic.\\
\item If $k=3$, then the ${\cal R}$-Steiner convexity is a convex geometry if and only if $G$ is $(P_4,G_1,G_2,G_3,G_4)$-free.
\end{enumerate}
\end{theorem}

To conclude this section, we briefly discuss a further convexity notion, also presented in~\cite{caceres-oellermann}, that combines the $3$-Steiner convexity and the so-called {\em $g^3$-convexity}. A set $S$ of vertices in a graph $G$ is {\em $g^3$-convex} if, for every pair $u,v\in S$ such that $u$ and $v$ are at distance at least $3$ in $G$, $I[u,v]\subseteq S$. (Here, the interval is taken over shortest paths.) A set $S$ is $g^3_3$-convex if it is both $g_3$- and $g^3$-convex. The convexity associated with $g^3_3$-convex sets is the {\em $g^3_3$-Steiner convexity}. 

\begin{theorem} {\em \cite{caceres-oellermann}}
The $g^3_3$-Steiner convexity of a graph $G$ is a convex geometry if and only if $G$ is $(house, hole, domino, $A$, gem, G_1, G_2, G_3, G_4, T_1, T_2)$-free. See Figures~\ref{fig:gem},\ref{fig:hhda},\ref{fig:replicated},\ref{fig:ts}.
\end{theorem}

\begin{figure}[h]
\begin{center}
\begin{tikzpicture}[scale=0.9]
\filldraw[black] (1,2) circle (2pt) node[anchor=south]{$T_1$};
\filldraw[black] (0,1) circle (2pt);
\filldraw[black] (1,1) circle (2pt);
\filldraw[black] (2,1) circle (2pt);
\filldraw[black] (3,1) circle (2pt);
\filldraw[black] (1,0) circle (2pt);
\draw (1,2) -- (2,1) -- (1,0) -- (0,1) -- (1,2);
\draw (0,1) -- (1,1) -- (2,1) -- (3,1);
\filldraw[black] (6,2) circle (2pt) node[anchor=south]{$T_2$};
\filldraw[black] (5,1) circle (2pt);
\filldraw[black] (6,1) circle (2pt);
\filldraw[black] (7,1) circle (2pt);
\filldraw[black] (8,1) circle (2pt);
\filldraw[black] (6,0) circle (2pt);
\draw (6,2) -- (7,1) -- (6,0) -- (5,1) -- (6,2);
\draw (5,1) -- (6,1) -- (7,1) -- (8,1);
\draw (6,2) -- (6,1);
\end{tikzpicture}
\caption{Graphs $T_1$ and $T_2$.}\label{fig:ts}
\end{center}
\end{figure}

\subsection{Toll convexity}

A walk $u_0u_1 \ldots u_{k-1}u_k$ is a \textit{tolled walk} if $u_0u_i \in E(G)$ implies $i = 1$ and $u_ju_k\in E(G)$ implies $j=k-1$. In~\cite{alcon}, tolled walks were used to characterize interval graphs. A graph $G$ is an {\em interval graph} if its vertices can be associated with intervals on the real line such that two vertices are adjacent if and only if the associated intervals intersect. The key references on the toll convexity are~\cite{alcon-et-al,dourado}.

Let ${\mathscr P}^t$ be the path system that maps a graph $G$ to its tolled walks. The convexity defined over such walks is called the {\em toll convexity} of $G$, and the associated convex sets are called {\em toll convex sets}.

If $I$ is an interval on the real line, let $R(I)$ and $L(I)$ be, respectively, the right and left endpoints of $I$. Given a family of intervals (or {\em interval model}\;) ${\mathscr I}=\{I_v\}_{v\in V(G)}$, associated with an interval graph $G$, we say that $I_a$ is an \textit{end interval} if $L(I_a)=\mathit{Min}\bigcup I_v$ or $R(I_a)=\mathit{Max}\bigcup I_v$, i.e., $I_a$ appears as the first or the last interval in ${\mathscr I}$. A vertex $a\in V(G)$ is an \textit{end vertex} if there exists some interval model of $G$ where $a$ is associated with an end interval. In addition, $a$ is an {\em end simplicial vertex} if $a$ is an end vertex and is simplicial. In~\cite{alcon-et-al}, two facts on extreme vertices in toll convex sets are given.

\begin{lemma}\label{lem:alcon} {\em \cite{alcon-et-al}}
If $v$ is an extreme vertex of a toll convex set $S$ of a graph $G$, then $v$ is simplicial in $G[S]$.
\end{lemma}

\begin{lemma}\label{lem:alcon2} {\em \cite{alcon-et-al}}
A vertex $v$ of an interval graph $G$ is an extreme vertex of a toll convex set $S\subseteq V(G)$ if and only if $v$ is an end simplicial vertex of $G[S]$.
\end{lemma}

In order to characterize the graphs with toll convexities that are convex geometries, we need to resort to a well-known characterization of interval graphs. Three vertices of a graph form an {\em asteroidal triple} if between any pair of them there exists a path that avoids the neighborhood of the third vertex.

\begin{theorem}\label{thm:lb} {\em \cite{lek-bol}}
A graph $G$ is an interval graph if and only if $G$ is chordal and contains no asteroidal triple.
\end{theorem}

Suppose that the toll convexity of $G$ is a convex geometry. Let us show first that $G$ is chordal. By the assumption, $V(G)$ (and every of its toll convex subsets) is the convex hull of its extreme vertices. Let $v$ be an extreme vertex of $V(G)$. By definition of extreme vertex, $V(G)\setminus\{v\}$ is toll convex in $G$. Thus, the toll convexity of $G-v$ is a convex geometry (note that any subset of $V(G)\setminus\{v\}$ is toll convex in $G-v$ iff it is toll convex in $G$). Using induction, $G-v$ is chordal. Since by Lemma~\ref{lem:alcon} $v$ is a simplicial vertex of $G$, it follows that $G$ is chordal as well (recall that a graph is chordal iff there is an ordering $v_1,\ldots,v_n$ of its vertices such that $v_i$ is simplicial in $G[\{v_i, v_{i+1}, \ldots, v_n\}]$, $1\leq i\leq n$~\cite{golumbic80}).

Now, let us show that $G$ contains no asteroidal triple. Suppose by contradiction that vertices $a,b,c$ form an asteroidal triple in $G$. Consider the walk $W$ from $a$ to $c$ formed by the concatenation of induced paths $P_{ab}$, from $a$ to $b$, and $P_{bc}$, from $b$ to $c$, such that $P_{ab}$ avoids $N[c]$ and $P_{bc}$ avoids $N[a]$. Note that $a$ is adjacent only to one vertex of $P_{ab}$ and to no vertices of $P_{bc}$. Likewise, $c$ is adjacent only to one vertex of $P_{bc}$ and to no vertices of $P_{ab}$. Thus, $W$ is a tolled walk from $a$ to $c$ passing through $b$. Let $Y=H(\{a,b,c\})$. Since $V(W)\subseteq Y$, $b$ is not an extreme vertex of $Y$. Analogously, $a$ and $c$ are not extreme vertices of $Y$. Since no other vertices of $Y$ are extreme vertices, this implies that $Y$ has no extreme vertices, a contradiction.

The precedent paragraphs tell us that if the toll convexity of $G$ is a convex geometry, then, by Theorem~\ref{thm:lb}, $G$ is an interval graph. To prove the converse, we use the following useful lemma:

\begin{lemma}~\label{lem:alcon3} {\em \cite{alcon-et-al}}
Let $G$ be an interval graph. Then, every vertex that is not an end simplicial vertex lies in a tolled walk between two end simplicial vertices of $G$.
\end{lemma}

Using Lemma~\ref{lem:alcon2}, Lemma~\ref{lem:alcon3}, and the fact that interval graphs are hereditary, we have:

\begin{theorem} {\em \cite{alcon-et-al}}
A graph $G$ is an interval graph if and only if the toll convexity of $G$ is a convex geometry.
\end{theorem}

\subsection{Weakly toll convexity}

A walk $u_0u_1 \ldots u_{k-1}u_k$ is a \textit{weakly toll walk} if $u_0u_i \in E(G)$ implies $u_i = u_1$ and $u_ju_k\in E(G)$ implies $u_j=u_{k-1}$. The concept of weakly toll walk is a relaxation of the concept of tolled walk (in the sense that every tolled walk is a weakly toll walk). In the previous subsection, tolled walks have been used to characterize interval graphs via convex geometries. Similarly, weakly toll walks are used as a tool to characterize proper interval graphs. A \textit{proper interval graph} is an interval graph that admits an interval model in which no interval properly contains another, or, equivalently, an interval model in which all the intervals have the same length. Roberts~\cite{roberts} proved that proper interval graphs are exactly the interval graphs containing no $K_{1,3}$ as an induced subgraph. The graph $K_{1,3}$ consists of four vertices $a,b,c,d$ and three edges $ab,ac,ad$.

Let ${\mathscr P}^w$ be the path system that maps a graph $G$ to its weakly toll walks. The convexity defined over such walks is called the {\em weakly toll convexity} of $G$, and the associated convex sets are called {\em weakly toll convex sets}. The basic references on the weakly toll convexity are~\cite{gutierrez-tondato,gutierrez-tondato-2}.

Lemmas~\ref{lem:alcon} and~\ref{lem:alcon2} have similar versions for weakly toll convex sets:

\begin{lemma}\label{lem:alcon4} {\em \cite{gutierrez-tondato}}
If $v$ is an extreme vertex of a weakly toll convex set $S$ of a graph $G$, then $v$ is simplicial in $G[S]$.
\end{lemma}

\begin{lemma}\label{lem:alcon5} {\em \cite{gutierrez-tondato}}
A vertex $v$ of a proper interval graph $G$ is an extreme vertex of a weakly toll convex set $S\subseteq V(G)$ if and only if $v$ is an end simplicial vertex of $G[S]$.
\end{lemma}

Using arguments similar to those used in the previous section, one can prove that if the weakly toll convexity of $G$ is a convex geometry, then $G$ is chordal and cannot contain asteroidal triples and induced subgraphs isomorphic to $K_{1,3}$. Hence, using the characterization of proper interval graphs in~\cite{roberts}, $G$ is a proper interval graph. Conversely, assume that $G$ is a proper interval graph. Then, it is an interval graph. This means that, by Lemma~\ref{lem:alcon3}, every vertex of $G$ that is not an end simplicial vertex lies in a tolled walk between two end simplicial vertices. Since every tolled walk is a weakly toll walk and, by Lemma~\ref{lem:alcon5}, end simplicial vertices of a proper interval graph are extreme vertices, we have:

\begin{lemma}~\label{lem:alcon6} {\em \cite{gutierrez-tondato}}
Let $G$ be a proper interval graph. Then, every vertex that is not an end simplicial vertex lies in a weakly toll walk between two end simplicial vertices of $G$.
\end{lemma}

\begin{theorem} {\em \cite{gutierrez-tondato}}
A graph $G$ is a proper interval graph if and only if the weakly toll convexity of $G$ is a convex geometry.
\end{theorem}

\subsection{$l^k$-convexity}

Let ${\mathscr P}^{\leq k}$ be the path system that maps a graph $G$ to its induced paths of length at most $k$. The convexity defined over such paths is called the {\em $l^k$-convexity} of $G$, and the associated convex sets are called {\em $l^k$-convex sets}. The $l^2$-convexity is also called {\em $P_3^*$-convexity}~\cite{araujo}.

Suppose that the $l^2$-convexity of $G$ is a convex geometry. It is easy to see that the extreme vertices of the $l^2$-convex set $V(G)$ are the simplicial vertices of $G$. Thus, if $G$ contains an induced cycle $C$ with at least four vertices, then $H(V(C))$ does not have simplicial (extreme) vertices, a contradiction. This means that $G$ is chordal. Now, we prove that $G$ is also a {\em cograph} (a graph that contains no $P_4$ as an induced subgraph -- see~\cite{corneil-et-al} for details). Let $P=x_0,\ldots,x_p$ be a maximum induced path of $G$, between two simplicial vertices of $G$. Such a path exists because $G$ is chordal~\cite{dirac}. Assume $p\geq 3$. The set $H(V(P))$ is an $l^2$-convex set of $G$ whose extreme vertices are $x_0$ and $x_p$. If there exists $z\in H(V(P))\setminus V(P)$ that lies in an induced path of length two between $x_0$ and $x_p$, $z$ must be adjacent to every $x_i$; otherwise, $P$ plus the edges $x_0z$ and $zx_p$ forms an induced cycle of $G$ with at least four vertices. Since $N[x_0]$ and $N[x_p]$ are cliques in $H(V(P))$,
\[H(V(P))=(N[x_0]\cap N[x_p])\cup\{x_0,x_p\}.\]
This implies $H(V(P))\neq H(\{x_0,x_p\})$, a contradiction. Thus $G$ is a chordal cograph.

Conversely, let $G$ be a chordal cograph, and let $S\subseteq V(G)$ be an $l^2$-convex set of $G$. Since $G[S]$ is chordal, every vertex in $S$ lies in an induced path between two simplicial vertices of $G[S]$. Such a path must have length two, otherwise $G$ contains an induced path of length at least four as an induced subgraph. Hence,

\begin{theorem} {\em \cite{araujo-sampaio,gpt}}
A graph $G$ is a chordal cograph if and only if the $l^2$-convexity of $G$ is a convex geometry.
\end{theorem}

We now analyze graphs with $l^3$-convexities that are convex geometries. A vertex $u$ is a {\em universal vertex} in a graph $G$ if $u$ is adjacent to every other vertex of $G$. An \textit{$n$-gem} is a graph $G_n$ such that: (i) $V(G_n)=\{x_0,\ldots,x_n,u_n\}$ ($n\geq 4$); (ii) $x_0,\ldots,x_n$ is an induced path; (iii) $u_n$ is a universal vertex. See Figure~\ref{fig:n-gem}. 

\begin{figure}[ht]
\tikzstyle{miEstilo}= [thin, dotted]
\begin{center}
\begin{tikzpicture}[scale=0.8]
\node[draw,circle] (1) at (-4,0) {$x_0$};
\node[draw,circle] (2) at (-2,0) {$x_1$};
\node[draw,circle] (3) at (2,0) {$x_n$};
\node[draw,circle] (4) at (0,2) {$u_n$};
\draw[miEstilo]  (2)--(3);
\draw (1)--(2);
\draw (2)--(4)--(3);
\draw (4)--(1);
\draw[miEstilo] (4)--(-1,0);
\draw[miEstilo] (4)--(0,0);
\draw[miEstilo] (4)--(1,0);
\end{tikzpicture}\end{center}
\caption{An $n$-gem.} \label{fig:n-gem}
\end{figure}

Suppose that $G$ contains an $n$-gem $G_n$ as an induced subgraph. We say that $G_n$ is {\em solved} if there exists in $G$ a $P_4$ connecting $x_0$ and $x_n$ that avoids $u_n$. In~\cite{gpt},  graphs with $l^3$-convexities that are convex geometries are characterized as follows:

\begin{theorem}\label{l3} {\em \cite{gpt}}
Let $G$ be a graph. The $l^3$-convexity of $G$ is a convex geometry if and only if the following conditions hold:

\begin{enumerate}

\item $G$ is chordal;

\item $\mathit{diam}(G)\leq 3$;

\item every induced $n$-gem $(n\geq 4)$ contained in $G$ is solved.
\end{enumerate}

\end{theorem}

The proof of the above theorem contains many details and will be omitted. The ``only if'' part can be generalized to prove the following result:

\begin{theorem} {\em \cite{gpt}}
Let $G$ be a graph and let $k\geq 2$. If the $l^k$-convexity of $G$ is a convex geometry, then $G$ is chordal and $\mathit{diam}(G)\leq k$.
\end{theorem}

Despite the above result, a complete characterization of graphs with $l^k$-convexities that are convex geometries, for arbitrary $k$, remains still as an open question.

A surprising fact is that the class of graphs characterized in Theorem~\ref{l3} is not hereditary. The following example is elucidative:

\begin{example}{\em
Let $G$ be the graph depicted in Figure~\ref{figure-example}. Note that the $l^3$-convexity of $G$ is a convex geometry. However, for $x\in\{2,5\}$, the $l^3$-convexity of $G-x$ is not a convex geometry, since:

\begin{itemize}
\item $V(G-x)$ is trivially an $l^3$-convex set of $G-x$,
\item $\mathit{ext}(V(G-x))=\{1,7\}$, and
\item $H(\{1,7\})=\{1,7\}\neq V(G-x)$.
\end{itemize}

\noindent This shows that the class of graphs with $l^3$-convexities that are convex geometries is not hereditary. As far as the authors know, this is the only practical example of a nonhereditary class of graphs with convex geometries of a certain type.}
\end{example}

\begin{figure}[ht]
\begin{center}
\begin{tikzpicture}[scale=0.8]
\node[draw,circle] (6) at (4,0) {$1$};
\node[draw,circle] (7) at (6,0) {$3$};
\node[draw,circle] (12) at (9,0) {$4$};
\node[draw,circle] (8) at (12,0) {$6$};
\node[draw,circle] (9) at (14,0) {$7$};
\node[draw,circle] (10) at (6,2) {$2$};
\node[draw,circle] (11) at (12,2) {$5$};
\draw (6)--(7)--(12)--(8)--(9)--(11)--(8);
\draw (12)--(10)--(11)--(12);
\draw (7)--(10)--(6);
\end{tikzpicture}\end{center}
\caption{The $l^3$-convexity of $G$ is a convex geometry.} \label{figure-example}
\end{figure}

\section{Other convexities}

In this section, we present some contributions to the study of convex geometries, some of which not associated with path systems.

\subsection{$P_3$ convexity}

Consider the path system that maps a graph $G$ to its paths of length two (not necessarily induced). The convexity defined over such paths is called the {\em $P_3$ convexity} of $G$~\cite{campos-et-al}, and the associated convex sets are called {\em $P_3$ convex sets}. Note that $S$ is a $P_3$ convex set if and only if every $x\notin S$ has at most one neighbor in $S$. Thus, the determination of the convex hull of a set in the $P_3$ convexity may be viewed as an infection process where noninfected vertices having two or more infected  neighbors become infected (for more details, see~\cite{campos-et-al} and references therein).

\begin{theorem}
The $P_3$ convexity of $G$ is a convex geometry if and only if $G$ is a forest of stars.
\end{theorem}

\begin{proof}
First, suppose that the $P_3$ convexity of $G$ is a convex geometry. If there are vertices $a,b,c \in V(G)$ inducing a $K_3$, then $a \in H(\{b,c\})$ and $b \in H(\{a,c\})$, which contradicts the anti-exchange property, since $\{c\}$ is a convex set. Therefore, $G$ is $K_3$-free and the endpoints of any edge of $G$ form a convex set. We also have that $G$ is $C_4$-free, because if there is an induced $C_4 = abcda$ in $G$, then $S = \{c,d\}$ is a convex set, $a \in H(S \cup \{b\})$, and $b \in H(S \cup \{a\})$, which contradicts the anti-exchange property. Now, suppose that $G$ contains an induced $C_5 = abcdea$ in $G$. Since $\{c,d\}$ is a convex set and $G$ is $C_4$-free, we have $H(\{c,e\}) = \{c,d,e\}$, because if there is $v \in H(\{c,e\}) \setminus \{c,d,e\}$ , then $\{c,d,e,v\}$ would be an induced $C_4$. Since $a \in H(S \cup \{b\})$ and $b \in H(S \cup \{a\})$, the anti-exchange property does not hold. Therefore, $G$ is $C_5$-free. Next, suppose that $G$ contains an induced $P_4 = abcd$. We have that $S = \{a,d\}$ is a convex set, because if there is a $P_3 = aed$ in $G$, then we would have that $be,ce \not\in E(G)$ (since $G$ is $K_3$-free) and that $abcde$ is an induced $C_5$, which is not possible. But then the anti-exchange property does not hold because $b \in H(\{a\} \cup S)$ and $c \in H(\{b\} \cup S)$. Thus, $G$ is $P_4$-free.

Since $G$ is $(K_3,C_4,P_4)$-free, we can say that $G$ is a cograph having no triangles or induced cycles with $4$ vertices. We know that every nontrivial connected cograph $G$ is the join\footnote{The {\em join} of two graphs $G_1$ and $G_2$ is the graph $G$ with $V(G)=V(G_1)\cup V(G_2)$ and $E(G)=E(G_1)\cup E(G_2) \cup \{xy\mid x\in V(G_1), y\in V(G_2)\}$.} of two graphs $G_1$ and $G_2$. Since $G$ is $K_3$-free, we have that the vertices of $G_i$ for every $i \in \{1,2\}$ induce an independent set; and since $G$ is $C_4$-free, we have that $|V(G_i)| = 1$ for at least one $i \in \{1,2\}$. Therefore, every connected component of $G$ is a star.

Conversely, consider that $G$ is a forest of stars. Since a star has at most one vertex with degree at least 2, $G$ satisfies the anti-exchange property. Thus the $P_3$ convexity of $G$ is a convex geometry.
\end{proof}

\subsection{Triangle path convexity}

We say that a path $P$ in graph $G$ is a {\em triangle path} if every two vertices of $P$ with distance greater than 2 in $P$ are non-adjacent in $G$. A set $S \subseteq V(G)$ is {\em tp-convex} if for any $x,y \in S$, every vertex belonging to some triangle path between $x$ and $y$ belongs to $S$. The convexity associated with the tp-convex sets of $G$ is the {\em triangle path convexity} of $G$~\cite{changat,dourado-sampaio}.

\begin{theorem}\label{thm:triangle-path}
The triangle path convexity of $G$ is a convex geometry if and only if $G$ is a forest.
\end{theorem}

\begin{proof}
First, suppose that the triangle path convexity of $G$ is a convex geometry. If there are vertices $a,b,c \in V(G)$ inducing a $K_3$, then $a \in H(\{b,c\})$ and $b \in H(\{a,c\})$, which contradicts the anti-exchange property, since $\{c\}$ is a convex set. Therefore, $G$ is $K_3$-free and the endpoints of any edge of $G$ form a convex set. We also have that $G$ is $C_k$-free for $k\geq 4$, because if there is an induced $C_k = v_1 \ldots v_kv_1$ in $G$ ($k \ge 4$), then $S = \{v_2,v_3\}$ is a convex set, $v_1 \in H(S \cup \{v_4\})$, and $v_4 \in H(S \cup \{v_1\})$, which contradicts the anti-exchange property. Hence, $G$ is acyclic.

Conversely, suppose that $G$ is a forest. As $G$ is $K_3$-free, the triangle path and the monophonic convexities of $G$ are identical. Since $G$ is a forest, $G$ is a chordal graph, which means, by Theorem~\ref{the:monophonic}, that both the monophonic and triangle path convexities of $G$ are convex geometries.	
\end{proof}

\subsection{All-path convexity}

A set $S\subseteq V(G)$ is \textit{ap-convex} if, for any $x,y\in S$, every vertex lying in a path between $x$ and $y$ belongs to $S$. The convexity associated with the ap-convex sets of $G$ is the \textit{all-path convexity} of $G$. A study on the all-path convexity can be found in~\cite{protti}.

It is easy to see that $v$ is an extreme vertex of an ap-convex set $S$ with $|S|>1$ if and only if $v$ is a pendant vertex (a vertex of degree one) in $G[S]$. 

\begin{theorem}\label{thm:all-path}
The all-path convexity of a graph $G$ is a convex geometry if and only if $G$ is a forest.
\end{theorem}

\begin{proof}
Suppose that the all-path convexity of $G$ is a convex geometry and $G$ contains a cycle $C$. Then, the vertices in $V(C)$ are non-pendant and $H(V(C))$ contains no extreme vertices, a contradiction. Thus $G$ is a forest. Conversely, if $G$ is a forest, then every non-pendant vertex of $G$ lies in a path between two pendant vertices, and this property holds for every induced subgraph of $G$. This means that every ap-convex set of $G$ is the convex hull of its extreme vertices.
\end{proof}

\subsection{Cycle convexity}

We say that a set $S\subseteq V(G)$ is {\em c-convex} if no vertex $u\notin S$ is adjacent to distinct $x,y\in S$ such that there is a path from $x$ to $y$ in $G[S]$; in other words, no $u\notin S$ is contained in a cycle $C$ with $V(C)\backslash S = \{u\}$. The convexity associated with the c-convex sets of $G$ is the \textit{cycle convexity} of $G$. See~\cite{araujo-campos,araujo-ducoffe} for important results on the cycle convexity.

Clearly, $x$ is an extreme vertex of a c-convex set $S$ if and only if there is no cycle $C$ in $G[S]$ containing $x$.  

\begin{theorem}\label{thm:cycle-conv}
The cycle convexity of a graph $G$ is a convex geometry if and only if $G$ is a forest.
\end{theorem}

\begin{proof}
If $G$ contains a cycle $C$, then $H(V(C))$ contains no extreme vertices. Thus, if the cycle convexity of $G$ is a convex geometry, then $G$ is a forest. Conversely, if $G$ is a forest, then every $S\subseteq V(G)$ is a c-convex set of $G$ and satisfies $\mathit{ext}(S)=S$, i.e., it is the convex hull of its extreme vertices.     
\end{proof}

It is interesting to note that Theorems~\ref{thm:triangle-path}, ~\ref{thm:all-path}, and~\ref{thm:cycle-conv} all apply to the class of forests. However, the triangle path, all-path, and cycle convexities do not coincide in general. If $G=C_n$ ($n\geq 4$), the families of tp-convex sets, ap-convex sets, and c-convex sets of $G$ are distinct.

\subsection{$\Delta$-convexity}

A set $S\subseteq V(G)$ is \textit{$\Delta$-convex} if no vertex $u\notin S$ has two adjacent neighbors in $S$, that is, forms a triangle with two vertices in $S$. The \textit{$\Delta$-convexity} of $G$ is the family formed by the $\Delta$-convex sets of $G$ (see~\cite{anand} for more details). It is easy to see that $x$ is an extreme vertex of a $\Delta$-convex set $S$ if and only if $N(x)$ is an independent set of $G[S]$. 

\begin{theorem}\label{thm:delta-conv}
The $\Delta$-convexity of a graph $G$ is a convex geometry if and only if $G$ is triangle-free.
\end{theorem}

\begin{proof}
Suppose that the $\Delta$-convexity of $G$ is a convex geometry. If $G$ contains a triangle $abc$, then $S=H(\{a,b,c\})$ contains no extreme vertices, since $N(x)$ is not an independent set for every vertex $x$ in $G[S]$. This is a contradiction. Conversely, if $G$ is triangle-free, then every $S\subseteq V(G)$ is a $\Delta$-convex set of $G$ and satisfies $\mathit{ext}(S)=S$, i.e., it is the convex hull of its extreme vertices.
\end{proof}       

\subsection{$\mathcal{F}$-free convexities}\label{sec:Ffree}

We now introduce the concept of \textit{$\mathcal{F}$-free convexity}~\cite{araujo-sampaio}. Let $H$ be a nontrivial graph. Given a graph $G$, we say that $S \subseteq V(G)$ is {\em $H$-free convex} if for every $S' \subseteq S$ with $|S'| = |V(H)| - 1$ and $x \in V(G)$, if the subgraph of $G$ induced by $S' \cup \{x\}$ is isomorphic to $H$, then $x \in S$. Given a family ${\cal F}$ of nontrivial graphs, we say that $S \subseteq V(G)$ is {\em ${\cal F}$-free convex} if $S$ is $H$-free convex for every $H \in {\cal F}$. The convexity associated with the ${\cal F}$-free convex sets of $G$ is the ${\cal F}$-free convexity of $G$.

\begin{theorem} \label{the:F-free}
The ${\cal F}$-free convexity of $G$ is a convex geometry if and only if $G$ is ${\cal F}$-free.
\end{theorem}

\begin{proof}
First, suppose that $G$ is ${\cal F}$-free. By definition, every subset of $V(G)$ is ${\cal F}$-free convex, which means that the anti-exchange property is valid. Therefore, the ${\cal F}$-free convexity of $G$ is a convex geometry.

Conversely, suppose that the ${\cal F}$-free convexity of $G$ is a convex geometry and $G$ contains a graph of ${\cal F}$ as an induced subgraph. Let $S \subseteq V(G)$ with minimum size such that $G[S]$ is isomorphic to a graph $H \in {\cal F}$. Let $x,y$ be distinct vertices of $S$ and write $S' = S \setminus \{x,y\}$. By the minimality of $S$, we have that $S'$ is ${\cal F}$-free convex. Note also that $x \in H(S' \cup \{y\})$ and $y \in H(S \cup \{x\})$, which means that the anti-exchange property does not hold. This is a contradiction. Therefore, $G$ is ${\cal F}$-free.
\end{proof}

\begin{corollary}
Let ${\cal F}$ be the family of odd cycles. Then, the ${\cal F}$-free convexity of $G$ is a convex geometry if and only if $G$ is bipartite.
\end{corollary}

\begin{corollary}
Let ${\cal F}$ be the family of graphs that can be obtained from $K_5$ or $K_{3,3}$ by subdivision of edges. Then, the ${\cal F}$-free convexity of $G$ is a convex geometry if and only if $G$ is planar.
\end{corollary}


The ${\cal F}$-free convexity defined by ${\cal F}=\{P_4\}$ is called \textit{$P_4^+$-convexity}. Recall that a {\em cograph} is a $P_4$-free graph. Given a graph $G$, we say that $S \subseteq V(G)$ is {\em $P^+_4$-convex} if, for every induced $P_4 = abcd$ in $G$, if $a,b,d \in S$, then $c\in S$. The convexity associated with the $P^+_4$-convex sets of $G$ is the $P_4^+$-convexity of $G$.

\begin{corollary}
The $P_4^+$-convexity of $G$ is a convex geometry if and only if $G$ is a cograph.
\end{corollary}

To conclude this section, we remark that the cycle convexity and the $\Delta$-convexity are ${\cal F}$-free convexities defined by taking ${\cal F}=\{C_k\mid k\geq 3\}$ and ${\cal F}=\{C_3\}$, respectively. Hence, Theorems~\ref{thm:cycle-conv} and~\ref{thm:delta-conv} can be viewed as corollaries of Theorem~\ref{the:F-free}. 

\section{Conclusions}

In this survey, we presented characterizations of several classes of graphs via convex geometries. The characterizations follow the following structure: A graph $G$ is in class ${\mathscr G}$ if and only if the ${\mathscr C}$-convexity of $G$ is a convex geometry. The table below summarizes the results.

An interesting question is to increase the table by including other important classes, such as distance-hereditary graphs and split graphs. Another line of research is to investigate convex geometries defined on digraphs. Finally, as far as the authors now, the $g^3$-Steiner convexities that are convex geometries have not been characterized yet.

\begin{center}
\begin{tabular}{ l | l }
${\mathscr G}$           & ${\mathscr C}$                 \\\hline
 chordal                 & monophonic                     \\
 Ptolemaic               & geodesic                       \\
 strongly chordal        & strong                         \\
 weak bipolarizable      & $m^3$                          \\
 Ptolemaic               & Steiner                        \\
 $(P_4,G_1,G_2,G_3,G_4)$-free & 3-Steiner                 \\
 (Theorem 10)                 & ${\cal R}$-Steiner        \\
 (Theorem 11)                 & $g^3_3$-Steiner           \\
 interval                & toll                           \\
 proper interval         & weakly toll                    \\
 chordal cograph         & $l^2$                          \\
 (Theorem 18)            & $l^3$                          \\
 forest of stars         & $P_3$                          \\
 forest                  & triangle path, all path, cycle \\
 triangle-free           & $\Delta$                       \\
 ${\mathcal F}$-free     & ${\mathcal F}$-free convexity  \\
 bipartite               & (Corollary 24)                 \\
 planar                  & (Corollary 25)                 \\
 cograph                 & $P_4^+$
\end{tabular}
\end{center}


\end{document}